\documentclass[reqno,10pt]{amsart}
\usepackage{amssymb}
\usepackage{amsmath}
\usepackage{amsthm}
\usepackage{pgf}
\usepackage{color}
\usepackage{graphicx}   
\usepackage{float}

\newtheorem{theorem}{Theorem}[section]

\newtheorem{proposition}[theorem]{Proposition}

\newcommand{\be}[1]{\begin{equation}\label{#1}}
\newcommand{\ee}{\end{equation}}

\newcommand{\Rz}{\mathbb{R}}

\newcommand{\DD}{{\rm D}}

\def\eps{\varepsilon}

\def\XXint#1#2#3{{\setbox0=\hbox{$#1{#2#3}{\int}$}
     \vcenter{\hbox{$#2#3$}}\kern-.5\wd0}}

\definecolor{applegreen}{rgb}{0.55, 0.71, 0.0}

\newcommand{\f}{\varphi}
\newcommand{\bfc}{{c}}

\usepackage{latexsym,amsfonts}
\usepackage{mathrsfs}
\usepackage{centernot}
\usepackage{paralist}

\usepackage{eso-pic}  
\usepackage{pifont}
\usepackage{fixmath} 

\usepackage{metalogo}
\usepackage{booktabs}
\usepackage{hyperref}
\makeatletter
\hypersetup{%
     pdfpagemode={UseOutlines},
     bookmarksopen,
     pdfstartview={FitH},
     colorlinks,
     linkcolor={blue},
     citecolor={red},
    urlcolor={red}
  }

\theoremstyle{definition}
\begingroup
\newtheorem{definition}[theorem]{Definition}
\newtheorem{remark}[theorem]{Remark}

\endgroup

\usepackage{metalogo}
\usepackage{booktabs}
\usepackage{hyperref}
\makeatletter
\hypersetup{%
     pdfpagemode={UseOutlines},
     bookmarksopen,
     pdfstartview={FitH},
     colorlinks,
     linkcolor={blue},
     citecolor={red},
    urlcolor={red}
  }

\usepackage{comment}

\begin{document}

\title[The geometry of $C_{60}$]{The geometry of $C_{60}$: \\
a rigorous approach via Molecular Mechanics}

\author{Manuel Friedrich}
\address[Manuel Friedrich]{Faculty of Mathematics, University of Vienna, Oskar-Morgenstern-Platz~1, A-1090 Vienna, Austria}

\author{Paolo Piovano}
\address[Paolo Piovano]{Faculty of Mathematics, University of Vienna, Oskar-Morgenstern-Platz~1, A-1090 Vienna, Austria}

\author{Ulisse Stefanelli}
\address[Ulisse Stefanelli]{Faculty of Mathematics,  University of Vienna, Oskar-Morgenstern-Platz~1, A-1090 Vienna, Austria $\&$ Istituto di Matematica Applicata e Tecnologie Informatiche ``E. Magenes'' - CNR, v. Ferrata 1, I-27100 Pavia, Italy}

\keywords{Fullerene, $C_{60}$, configurational energy minimization,
  local stability}

\begin{abstract} 
Molecular Mechanics describes molecules as particle
configurations interacting via classical potentials. 
These {\it configurational energies} usually
consist of the sum of different phenomenological terms which are
tailored to the description of specific bonding
geometries. This approach is followed here to model the
fullerene $C_{60}$, an allotrope of carbon corresponding to a specific
hollow spherical structure  of sixty atoms.  We rigorously address different
modeling options and advance a set of minimal
requirements on the configurational energy able to deliver an
accurate prediction of the fine three-dimensional geometry of
$C_{60}$ as well as of its remarkable stability.  In particular,
the experimentally observed truncated-icosahedron structure with two
different bond lengths is shown to be  a strict local minimizer. 
\end{abstract}

\subjclass[2010]{82D25.} 
\maketitle

\pagestyle{myheadings}

\section{Introduction}

 The molecule  $C_{60}$ is an allotrope of carbon formed by $60$ atoms  sitting  at the
vertices of a truncated icosahedron. Theoretically discussed in 
\cite{Osawa}  and \cite{Botchvar}, its serendipitous experimental discovery
in 1985 lead to the attribution of the 1996 Nobel Prize in
Chemistry to Curl, Kroto, and Smalley \cite{Kroto85,Kroto92}. This truly
remarkable result paved the way for extending the up-to-then known
allotropes, namely graphite,
diamond, and  amorphous carbon, to a whole new class of molecules consisting
of hollow carbon cages, balls, ellipsoids, and nanotubes. The
resemblance of $C_{60}$ with the geodesic domes by the  American
architect  Buckminster 
Fuller  has brought to name these molecules {\it fullerenes}.

 Fullerenes  have attracted an immense deal
of attention. The identification of their three-dimensional structure,
the study of their chemical properties among which aromaticity,
solubility, and electrochemistry, and their application in medicine
and pharmacology have
developed into the new branch of {\it Fullerene Chemistry}. A
central question  concerning fullerenes is their  {\it
  stability}  
\cite{Kroto87}, either from the thermodynamic, the electrochemical,
or the mechanical standpoint. Stability is believed to be the key
factor in explaining why just a few fullerene isomers out of a theoretically
predicted wide variety have been
actually revealed. Among these the fullerene $C_{60}$ is remarkably
stable   and considerable amounts of   these 
molecules have been detected in interstellar space,  despite the  harsh radiation environment \cite{ABBS}.

The aim of this paper is to provide a rigorous discussion of the
geometric structure and the
stability properties of the $C_{60}$ molecule. This analysis is set
within  the variational
frame of {\it Molecular Mechanics} \cite{Allinger,Lewars,RC}. This
consists in 
modeling molecular
configurations in terms of classical mechanics:  atomic relations are
described by
classical interaction potentials between atomic positions. Although far from the quantum
nature of molecular  bonding,  this approach has proved computationally
effective, especially in the case of large molecules, bringing indeed 
to the award of the 2013 Nobel
Prize in Chemistry to 
Karplus, Levitt, and Warshel. We shall express the {\it energy} of a
carbon configuration as
\begin{equation}\label{e} E  = E_{\rm bond} + E_{\rm angle} + E_{\rm dihedral} + E_{\rm  nonbond }.
\end{equation}
In the latter, $E_{\rm bond}$ describes two-body interactions, it is 
short-ranged, and favors some specific bond length, here
normalized to $1$. The term $E_{\rm angle}$ is a three-body
interaction energy instead \cite{Brenner90,Stillinger-Weber85,Tersoff},  favoring the formation of $2\pi/3$ or $4\pi/3$ angles between
first-neighbor bonds.  This corresponds to  the so-called $sp^2$-orbital hybridization of
carbon atoms, determining indeed the geometry of
approximately flat, locally two-dimensional carbon structures, such as
graphene and nanotubes.  Note that reducing to pure $sp^2$ hybridization to
describe the  truly
three-dimensional nature of $C_{60}$ is questionable. Still, this
simplification delivers the correct geometry of the molecule and 
possible extensions of this perspective are reported in Remark
\ref{remark}. The term $E_{\rm dihedral}$ is a four-body
contribution, favoring planarity of the bonds at a given atom.
  Finally, the term $E_{\rm
   nonbond }$ represents  {\it nonbonded}
interactions. These may include van der Waals attraction, steric repulsion, and electrostatic
effects.  

Our focus is to identify a minimal set of assumptions on $E$ delivering the
local minimality of the correct geometric structure of $C_{60}$: the
sixty atoms sit at the intersections of the edges of  an 
 icosahedron with a sphere with the same center. This results in a football-like geometry
consisting of twelve equal
regular planar pentagons and twenty equal planar hexagons. We call $\mathcal X$ all such
truncated-icosahedral configurations and remark that they are uniquely
 determined  (up to isometries) by specifying the lengths $a$ of the
side shared by two hexagons and the length $b$ of the sides of the
pentagons  (see Figure \ref{csixty}).  The corresponding configuration is indicated by
$X_{a,b}$. These two lengths are indeed different for the $C_{60}$ molecule: nuclear-magnetic-resonance experiments provide values of $a=1.40\pm 0.015$ \AA
\, and  $b=1.45\pm 0.015$ \AA, respectively \cite{Yannoni-etal}. 

\begin{figure}[tp] 
\begin{center}
\includegraphics[scale=0.4]{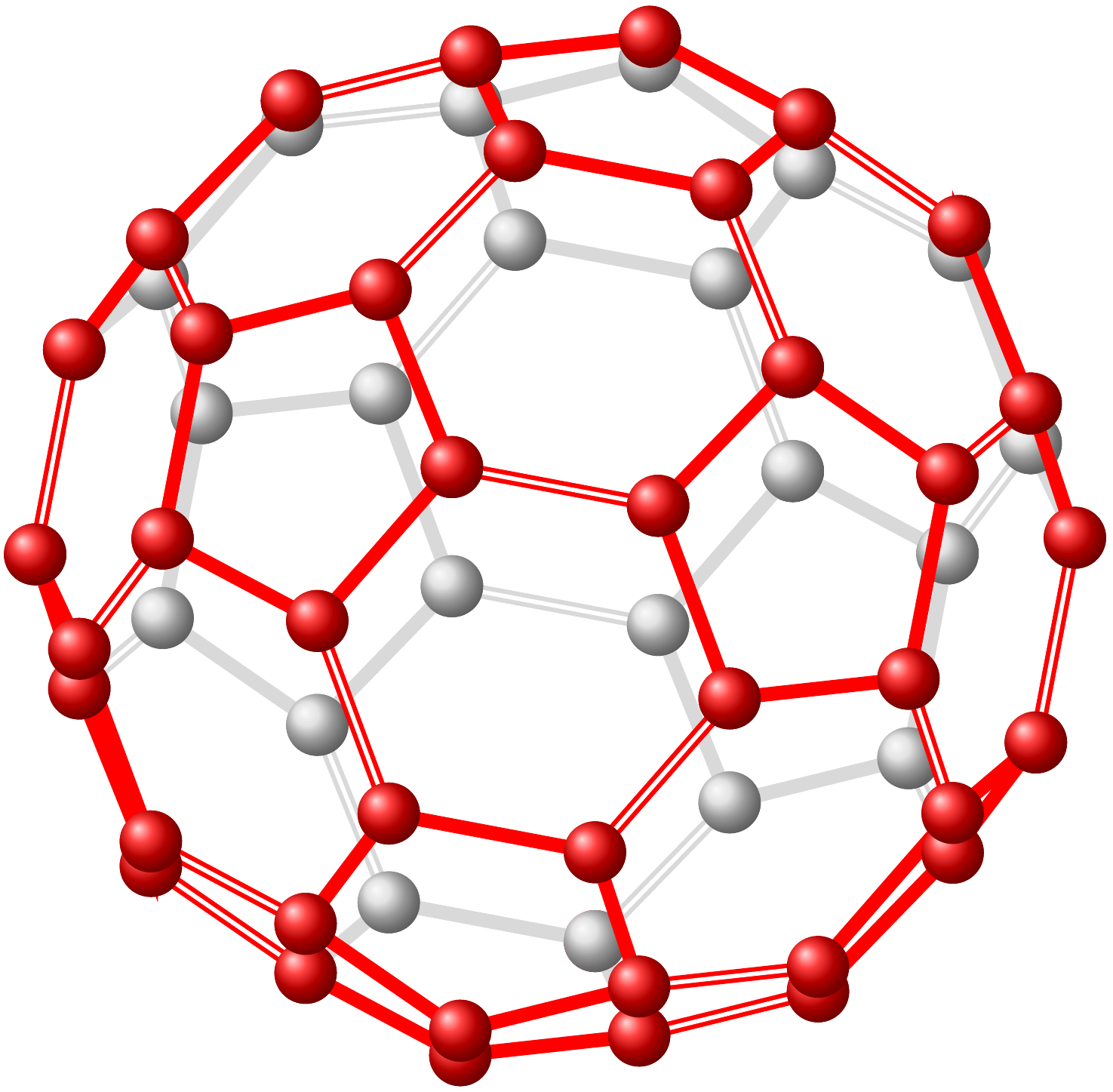}
 \caption{The geometry of $C_{60}$. The bonds shared by two hexagons are here indicated by double lines.} 
\label{csixty}
\end{center}
\end{figure}

The investigation of the structure of $C_{60} $ via
variational methods has been initiated in
\cite{Mainini-Stefanelli12,Stefanelli16}, where,
under suitable convexity assumptions, the energy $E_{\rm bond}+E_{\rm
  angle}$ is proved to be locally minimized by $X_{1,1}$. The  short-rangedness 
of this particular
energy form induces  this  
local minimizer  to have all bonds of length $1$, 
 which is indeed not reflecting the fine geometry of
the $C_{60}$ with two distinct bond lengths. In order to take long-range effects into account one is
tempted to consider the energy $E_{\rm bond}+E_{\rm
  angle} + E_{\rm  nonbond }$ instead. This has the effect of
bringing  (at least) second-neighbors into the picture, hence potentially
distinguishing between  bonds of type $a$ and $b$.  The
addition of the  nonbonded-interaction  term, however, induces a shortening of
second-neighbor bonds and a key conclusion of our paper is the observation that without additional assumptions on the energy local minimality of any configuration $X_{a,b}$ is  eventually prevented. Our main positive result is then that the inclusion of a dihedral term in the
energy, namely $E = E_{\rm bond} + E_{\rm angle} + E_{\rm dihedral} +
E_{\rm  nonbond } $, restores the  icosahedral  symmetry and entails the local minimality
of a 
configuration $X_{a^*,b^*}$ with $a^* <
 b^*$. 

In the following we critically review the effect of single terms
in $E$, by providing an accurate formalization of the above discussion. With
respect to the original {\it computational} nature of
Molecular Mechanics this {\it rigorous} approach seems  unprecedented ,
contributing indeed a novel justification of the variational
perspective and a way of validating specific modeling choices. Indeed, a variety of
different molecular mechanical codes \cite{Brooks83,Clark89,Gunsteren87,Mayo90,Weiner81}
have been presented, corresponding to different
phenomenological choices for the single terms in $E$ (as well as for possible
additional effects, not included in our analysis). A by-product of our
results is hence the 
cross-validation of these choices in view of their capability of
describing the actual geometry of $C_{60}$.  Let us briefly
  mention that  the modeling options discussed in this article are
  also consistent with the characterization of the geometry and
  stability of other carbon structures such as graphene or  the
  fullerene  $C_{20}$. Moreover,  this  variational
  approach has  proved effective to describe  a wider class of  carbon structures \cite{Stefanelli16}, including e.g.  carbon nanotubes \cite{MMPS1,MMPS2}.
  
    Before moving on, we would like to contextualize the results
  of this paper with respect to the available literature. 
Our analysis is related to the  classical   {\it
  Crystallization problem}, which consists in characterizing crystals
at zero temperature as periodic ground states  of suitable
configurational energies that include two- and three-body interaction
terms. 

In one space dimension,  the reader is referred  with no claim of
completeness to
\cite{Blanc-et-al02,Gardner-Radin79,Hamrick-Radin79,Radin83,Ventevogel78,Ventevogel-Nijboer79,Ventevogel-Nijboer79bis,
  Wagner83} for a collection of results proving or disproving, under
different choices for the energy, the minimization property of an
equally spaced configuration of atoms and its stability with respect
to perturbations.

Ground states in two dimensions have been proved to be subsets of the
triangular lattice under pure two-body interactions in
\cite{Heitman-Radin80,Radin81,Wagner83} for specific potentials. The considerably more involved case of
Lennard-Jones-like potentials has been analyzed  in
\cite{Theil06} as the number
of atoms of the configuration tends to infinity.  The  hexagonal case
is addressed by including in the energy a  
three-body interaction term favoring 
wells at $2\pi/3$ and $4\pi/3$ angles both in the finite
crystallization case \cite{Mainini-Stefanelli12} and in  the thermodynamic
limit \cite{E-Li09}. The recent \cite{Smereka15} obtains  a  hexagonal
lattice in the thermodynamic limit under the effect of an energy
favoring $\pi$  angles  instead. Eventually, the case of the square
lattice is tackled in \cite{MPS,MPS2}. Here the energy favors
$\pi/2$, $\pi$, and $3\pi/2$ bond angles.

The only three dimensional crystallization result presently available
is in \cite{Flatley-Theil12}, where a face-centered cubic lattice is
recovered as the thermodynamic limit under pairwise and three-body
interactions favoring $\pi/3$ bond  angles, see also 
\cite{Flatley-Taylor-Tarasov-Theil12}. All the mentioned results concern
the zero-temperature setting. Finite temperatures have been tackled in
the one dimensional case only  \cite{Jansen-atal}.  We refer the
reader to \cite{BL} for an extended review on this topic. 

In contrast with the classical  crystallization problem, we are not
concerned  here with ground-state characterization but rather with the
analysis of the $C_{60}$ configuration in terms of its stability.
Note that various concepts  of crystal stability are available. We refer
the reader to  \cite{elliott-atal1} for a discussion of the connections between phonon-stability,
homogenized-continuum stability, and Cauchy-born stability in the case
of three-dimensional crystals and  \cite{sfyris} for an application at
the continuum level for free-standing graphene. The validity of the
Cauchy-Born assumption for crystalline solids has been also discussed
in \cite{e-ming}  and 
\cite{Flatley-Theil02}. All these different stability notions are
qualified via the specification of the corresponding admissible
perturbations. In this regard, our stability notion seems to be the strongest since all small perturbations of atomic positions are allowed.

The plan of the paper is the following. We formalize our setting  and we state our main  results in Section \ref{main results
  section}. We also provide a classification of all
  modeling  options  by exactly characterizing the  cases in
  which $C_{60}$ can be identified as a local minimizer of the energy  (cf. Table \ref{table}).   We report in
Subsection \ref{discussion} a discussion of our assumption frame with
respect to various phenomenological
potentials from the literature.  The proof of the   results is then developed in  Sections
\ref{convexity}-\ref{counterexample section}.  More precisely, the
description of the geometry  of $C_{60}$  is addressed in
Section \ref{convexity} and  its   stability under the
presence of a dihedral term are contained in Section \ref{stability
  section}. Afterwards, in Section \ref{counterexample section} we
provide some counterexamples to stability  in absence  of a
dihedral term. These  consist in rotating one pentagonal facet or
simultaneously moving the vertices of a pentagonal facet towards the
center of the cage.

\section{Modeling and main results}\label{main results section}

The focus of this section is on introducing the relevant notation and stating
the main  results.

\subsection{Mathematical setting and modeling options}\label{model}
Let $X=\{x_1,\dots,x_{60}\} \in \Rz^3$  indicate a general {\it
  configuration} of sixty atoms in three-dimensional space and let \linebreak
$E: (\Rz^3)^{60}\to \Rz$ be a given {\it configurational energy}
\eqref{e}. The fundamental principle of material objectivity imposes
$E$ to be invariant under rotations and translations. As such, all the following statements have to be intended {\it up to isometries}, unless otherwise specified.

 We shall introduce some specific structure of the terms in \eqref{e}
 by modeling  the basic chemistry of $sp2$-covalent bonding in carbon \cite{Stillinger-Weber85,Tersoff}, namely the specific bonding mode of $C_{60}$.  We  define the two-body  interaction  term $E_{\rm bond}$ as
 \begin{equation*}\label{twobody}
E_{\rm bond}(X):=\frac12\sum_{(i,j)\in\,\mathcal{N}_1(X)}v_{\rm bond}(|x_i-x_j|)
\end{equation*}
where the index set $\mathcal{N}_1(X)$ indicates {\it first neighbors}
and is defined as
$$\mathcal{N}_1(X):=\{(i,j)\,:\,\textrm{ $|x_i-x_j|<\sqrt{2}$}\}.$$
 The potential $v_{\rm bond}:[0,\infty) \to [-1,\infty)$ is assumed to be smooth in the closure of a small open neighborhood $I_{\rm bond}$ of   $1$, 
\begin{align}\label{H2}
v_{\rm bond}(\ell)=-1 \ \text{iff} \ \ell=1,\quad\textrm{and}\quad v_{\rm bond}''(\ell) >0 \text{  for all $\ell \in \overline{I}_{\rm bond}$}.
\end{align}  

This basic
assumptions corresponds to the fact
that covalent bonds in carbon atoms are characterized by some
reference bond length, here normalized to $1$. The choice of the
cut-off value
$\sqrt{2}$ in the definition of first neighbors  is discretional,
yet suggested by the planar case of graphene 
\cite{Mainini-Stefanelli12}.
 In the following we shall also use the notation
$$\mathcal{N}_1(x_i):=\{(i,j) \ : \ \textrm{$j\in\{1,\dots, 60\}$ and $(i,j)\in\mathcal{N}_1(X)$}\}$$
for the set of first neighbors of the atom $x_i\in X$ and denote the
tuples of  lengths of the covalent bonds shared by $x_i$ by 
$$\mathcal{B}_1(x_i):=\{|x_i-x_j| \ : \ (i,j)\in\mathcal{N}_1(x_i)\}.$$

The energy $E_{\rm angle}$ represents three-body interactions and is defined by
\begin{equation*}\label{threebody}
E_{\rm angle} (X):=\frac{1}{2}\sum_{(i,j,k)\in \mathcal{T}(X)}v_{\rm angle}(\alpha_{i j k})
\end{equation*}
where  the index set $\mathcal{T}(X)$ is given by
$$\mathcal{T}(X):=\{(i,j,k)\,:\,\textrm{$i\neq k$, $(i,j)\in\mathcal{N}_1(X)$ and $(j,k)\in\mathcal{N}_1(X)$}\},
$$
and $\alpha_{i j k}$ denotes the angle determined by the segments $x_i - x_j$ and $x_k - x_j$ (choose anti-clockwise orientation, for definiteness). The potential $v_{\rm angle}: [0,2\pi] \to [0,\infty)$ is
symmetric with respect to $\pi$, attains its minimum value $ 0$  only
at $2\pi/3$ and $4\pi/3$, and is strongly convex in a small closed
neighborhood $I_{\rm angle}$ of  $[3\pi/5,2\pi/3]$,  i.e., 
\begin{align}\label{H3}
\theta \mapsto v_{\rm angle}(\theta) - \lambda_{\rm angle}|\theta|^2 \text{ is convex on $I_{\rm angle}$ for some $\lambda_{\rm angle}>0$}.
\end{align}
These properties will be assumed throughout the paper and model the fact that $sp2$-hybridized orbitals tend to form $2\pi/3$ bond angles \cite{Tersoff}.
The index set
$$\mathcal{A}(X):=\{\alpha_{ijk} \ : \ \textrm{$(i,j,k)\in\mathcal{T}(X)$ and $\alpha_{ijk}\leq\alpha_{kji}$}\}$$
will indicate {\it active angles} of the configuration $X$ while we will denote by 
$$\mathcal{A}(x_j):=\{\alpha_{ijk}\in \mathcal{A}(X)\ : \ \textrm{for some $i,k= 1,\dots,60$}\}$$
the tuple of the active angles at $x_j$. In the following, we will also make use of an alternative three-body energy term $E_{\rm kink}$ of the form of $E_{\rm angle}$, namely
\begin{equation*}\label{threebodyk}
E_{\rm kink} (X):=\frac{1}{2}\sum_{(i,j,k)\in \mathcal{T}(X)}v_{\rm kink}(\alpha_{i j k})
\end{equation*}
where $v_{\rm kink}$ fulfills the same  assumptions  as $v_{\rm angle}$
and is additionally differentiable in a small left neighborhood of
$2\pi/3$ with 
\begin{equation}
\lim_{\theta \uparrow 2\pi/3} v'_{\rm
  kink}\left(\theta\right)< 0.\label{kink}
\end{equation}
Note  that $v_{\rm kink}$ is not differentiable at
$2\pi/3$ and $4\pi/3$ where indeed its graph has a {\it kink}.
This is a mathematical assumption which has no
explicit chemical justification. Still, such   a nondifferentiable case is
surprisingly the only one allowing to prove that two-dimensional
minimizers of $E_{\rm bond}+E_{\rm kink}$ are indeed subsets of the
regular hexagonal lattice \cite{Mainini-Stefanelli12}. Note that such
{\it kink assumptions} arise in all finite crystallization results to
date. In particular, they have been considered in connection with the
two-dimensional triangular lattice and the square lattice as well
  \cite{Heitman-Radin80,MPS,MPS2,Radin81,Wagner83}, and see \cite{DPS1,DPS2} for related results.  As such, we
believe the discussion of the term $E_{\rm kink}$ to bear some
relevance.  

The four-body dihedral term $E_{\rm dihedral}$ is defined by
\begin{equation*}\label{fourbody}
E_{\rm dihedral}(X):=   \eta'\sum_{i=1}^{60} v_{\rm dihedral}(\alpha^1_i,\alpha_i^2,\alpha_i^3)
\end{equation*}
for a potential $v_{\rm dihedral}: [0,2\pi)^3\to[0,\infty)$ for all
configurations $X=\{x_1,\dots,x_{60}\}$ such that
$\#\mathcal{A}(x_i)=3$ for all $i=1,\dots,60$,  where $\mathcal{A}(x_i) = \lbrace  \alpha^1_i,\alpha_i^2,\alpha_i^3 \rbrace$. The constant $
\eta'>0$ will be chosen to be suitably small, corresponding indeed to
the smallness of four-body energy effects w.r.t. two- and three-body
energy contributions. We will assume $v_{\rm dihedral}$  to be smooth,  symmetric in its variables and to satisfy
 \begin{align}\label{v4hp}
  &\frac{\rm d }{{\rm d} \varphi} v_{\rm dihedral}(3\pi/5, \varphi,\varphi)\Big|_{\varphi = 2\pi/3}<0.
 \end{align}
The effect of the term $E_{\rm dihedral}$ is that of favoring the
planarity of active bonds at each atom. This again corresponds to
the local bonding geometry of $sp2$ covalent bonding~\cite{Tersoff}.

Eventually,  nonbonded  interactions are included in the energy by considering the term $E_{\rm     nonbond }$ defined by
\begin{equation}\label{H2s}
E_{\rm    nonbond }(X):= \frac{\eta}{2}     \,\sum_{i=1}^{60} \sum_{(j, i,k)\in\mathcal{T}(X)}v_{\rm  nbd }(|x_k-x_j|)
\end{equation}
for a smooth function $v_{\rm  nbd }: [0,\infty)\to[-1,\infty)$
increasing in a small  neighborhood $I_{\rm  nbd }$ of
$[2\sin(3\pi/10),\sqrt{3}]$. The constant $\eta >0$ will be chosen to
be suitably small later on, reflecting indeed the different relevance of  the effects of 
first and second neighbors  in covalent bonding. Note that in
\eqref{H2s}  the potential $v_{\rm  nbd }$ is evaluated over
{\it second neighbors} only, namely atoms corresponding to pairs
$$\mathcal{N}_2(X):=\{(i,k)\ :\ \textrm{$(i,j,k)\in\mathcal{T}(X)$ for
  some $j=1,\dots,60$}\}.$$ 
In particular, we
assume  nonbonded-interaction  effects to be negligible except for  
second neighbors.
We also denote by 
$$\mathcal{B}_2(x_j):=\{|x_i-x_k| \ : \  (i,j,k)\in\mathcal{T}(X)\}$$
the tuple of distances to second neighbors related to the atom $x_j\in X$.
 All the above assumptions on the potentials $v_{\rm bond}$, $v_{\rm angle}$, $v_{\rm kink}$,
  $v_{\rm dihedral}$, and $v_{\rm  nbd }$ are tacitly assumed
  throughout the paper.

In the following we discuss the effect of the various terms in \eqref{e}. For the sake of definiteness we introduce here a more specific notation for the configurational energy by letting 
\begin{align}\label{energyE}
E_\bfc(X)&:=E_{\rm bond}(X)+ c_{\rm angle}E_{\rm angle}(X)+c_{\rm kink}E_{\rm kink}(X)\nonumber\\
& \ \ \ \ +c_{\rm dihedral}E_{\rm dihedral}(X) + c_{\rm  nbd }E_{\rm    nonbond }(X).
\end{align}
The constants $c_{\rm angle},\, c_{\rm kink},\, c_{\rm dihedral}$, and
$c_{\rm  nbd }$ take values in $\{0,1\}$ and are hence intended to
switch on and off the different energy terms. Correspondingly,
different energies in \eqref{energyE} will be indicated by different
vectors $$\bfc=(c_{\rm angle}, c_{\rm kink}, c_{\rm dihedral},c_{\rm
   nbd }) \in \{0,1\}^4.$$  
In the following $c_{\rm angle}$ and $ c_{\rm kink}$ are never
simultaneously equal to $1$, since the two energies $E_{\rm angle}$
and $E_{\rm kink}$ indeed correspond to the same three-body contribution but
distinguish the case without or with the kink,
respectively. Moreover, it seems natural to consider the four-body
contribution $E_{\rm dihedral}$ only in the case where also three-body
terms are present, namely either for $c_{\rm angle}=1$ or $ c_{\rm
  kink}=1$. Under these restrictions, the discussion of all
possible vectors $\bfc$ of coefficients reduces to exactly
ten cases, all of which are addressed in Theorem \ref{maintheorem}, see
also Table \ref{table}.

\subsection{Main results} Among all configurations a specific subclass $\mathcal{X}$ of {\it
  objective} \cite{James} configurations with icosahedral symmetry
will play a  major  role. These correspond to truncated icosahedra with
two possibly distinct bond lengths and are defined as follows.
\begin{definition}[Icosahedral  Configurations]\label{familyc60}
The set $\mathcal{X}$ is the family of configurations
$X_{a,b}:=\{x_1,\dots,x_{60}\}$, $a,b\in I_{\rm bond} $,
corresponding to the intersections of the edges of a regular icosahedron
with a sphere with the same center,  where $\mathcal{B}_1(x_i)=\{a,b,b\}$
for all $i=1,\dots,60$.
\end{definition}
The set $\mathcal{X}$ is hence a two parameter family of configurations: by connecting first neighbors of $X_{a,b}$ by a straight segment one obtains a polyhedron with twelve regular pentagonal facets with side $b$ and twenty hexagonal facets with three sides of length $a$ and three of length $b$, alternating  (see Figure \ref{csixty}).  In particular, $X_{a,a}$ is a regular truncated icosahedron with side $a$ and we have that, for all $x_i \in X_{a,b}$, 
\[
  \mathcal{A}(x_i)=\{3\pi/5,2\pi/3,2\pi/3\}
\]
  and
\[\mathcal{B}_2(x_i)=\{p,h,h\}, 
\]
where
\begin{equation}\label{second bonds}
p:=2b\sin(3\pi/10)\quad\textrm{and}\quad h:=\sqrt{a^2+b^2 +ab}.
\end{equation}  

Note that the {\it angular part} of the energy $E_\bfc$ given by 
$$c_{\rm angle} E_{\rm angle}+c_{\rm kink} E_{\rm kink}+c_{\rm dihedral} E_{\rm dihedral} $$
is constant over $\mathcal{X}$ so that the minimization of $E_\bfc$ on $\mathcal{X}$ reduces to the two-dimensional problem
\begin{align}\label{reduced problem}
\min_{a,b \in I_{\rm bond}} \big( E_{\rm bond}(X_{a,b}) + c_{\rm  nbd }
E_{\rm    nonbond }(X_{a,b}) \big).
\end{align}
 In Section \ref{convexity} we will show that the latter energy is convex with respect to $(a,b)$ whenever $\eta$ is chosen
small enough.

\begin{table}[h]\centering
\begin{tabular}{|c|c|c|c||c||c|}
    \hline
    $c_{\rm angle}$ & $c_{\rm kink}$ & $c_{\rm dihedral}$ & $c_{\rm  nbd }$
    & local minimizer & Ref. in Theorem \ref{maintheorem}\\
\hline\hline 
1& 0 & 0 & 0 & $X_{1,1}$ & Assertion 3.1  \\
    \hline
    0& 1 & 0 & 0 & $X_{1,1}$ & Assertion 3.1  \\
\hline
    1 &  0 & 1 & 0 & $X_{1,1}$ & Assertion  3.1 \\
\hline
     0 &  1 & 1 &  0& $X_{1,1}$ & Assertion  3.1  \\
\hline\hline

    0& 1 & 0 & 1 & $X_{a^*,b^*}$ & Assertion  3.2  \\
    
    \hline
    0& 1 & 1 & 1 & $X_{a^*,b^*}$ & Assertion 3.3  \\
    \hline
    1& 0 & 1 & 1 & $X_{a^*,b^*}$ & Assertion 3.3  \\
    \hline
\hline
    
    0& 0 & 0 & 0 & not in $\mathcal{X}$ & Assertion 4.1  \\
\hline
    0 & 0 & 0 & 1 & not in $\mathcal{X}$ & Assertion 4.2  \\
    \hline
    1& 0 & 0 & 1 & not in $\mathcal{X}$ & Assertion 4.2  \\
    \hline
  \end{tabular}
\vspace{3mm}
\caption{ Illustration of ten possible cases in  Theorem
  \ref{maintheorem}. The first two from the top are already considered
in \cite{Mainini-Stefanelli12,Stefanelli16}.  Cases 5-7 are
the core result of the paper: the fine geometry of $C_{60}$ with two
different bond lengths can be modeled by allowing  nonbonded
interactions  in combination with a kink-angle or a dihedral term. By providing explicit perturbations we will see that $X_{a^*,b^*}$ is not stable in the last three cases.}  \label{table}
\end{table}

Out tenet is that the analysis of the above two-dimensional problem
actually delivers information on the 180-dimensional problem $\min
E_\bfc$.  In particular, some  choices  of the vector $\bfc$ entail the
local stability of specific configurations in
$\mathcal{X}$ with respect to all (small) perturbations in $(\Rz^3)^{60}$.   

In order to investigate such stability, let us
introduce perturbations of configurations $X_{a,b}= \{x_1,\dots,x_{60}\} \in
\mathcal{X}$ with respect to  the  energy $E_\bfc$  as
 \begin{equation}\label{perturbations}
 \mathcal{P}(X_{a,b}):=\{\{x_1',\dots,x_{60}'\} \ :\ \textrm{$x_i':=x_i +\delta x_i$ with $|\delta x_i|<\delta_0$ for all $i=1,\dots,60$}\}\end{equation}
where $\delta_0>0$ is chosen to be small enough. In particular, we
choose $\delta_0$ so small that every $P = \lbrace x_1',\ldots, x_{60}' \rbrace \in\mathcal{P}(X_{a,b})$ is such that 
\begin{equation}
\label{perturbedbonds}\mathcal{B}_1(x_i')=\{a_i,b^1_i,b^2_i\}
\end{equation}
 with $a_i, b_i^1,b_i^2\in I_{\rm bond}$, 
\begin{equation*}
\label{perturbedangles}\mathcal{A}(x_i')=\{\theta_i,\varphi_i^1,\varphi_i^2\}\end{equation*}
 with $\theta_i\in \Theta_{\eps}$ and $\varphi_i^1,\varphi_i^2\in
 \Phi_{\eps}$ where $\Theta_{\eps}:=(3\pi/5-\eps,3\pi/5+\eps)$and  $
 \Phi_{\eps}:=(2\pi/3-\eps,2\pi/3+\eps)$ for every every $x_i'\in P$,
 and for $\eps$ small depending on $\delta_0$ and
\begin{equation}
\label{perturbed2bonds}\mathcal{B}_2(x_i')=\{p_i,h_i^1,h_i^2\}\end{equation}
  with $p_i, h_i^1, h_i^2 \in I_{\rm  nbd }$,
  where
  \begin{equation}
\label{perturbed2bonds2}
p_i = \sqrt{(b^1_i)^2 + (b_i^2)^2 - 2b^1_i b^2_i \cos(\theta_i)}\quad\textrm{and}\quad h_i^j = \sqrt{a_i^2 + (b_i^j)^2 - 2a_i b_i^j \cos(\varphi_i^j)}
\end{equation}
for $ j=1,2$ (see Figure \ref{topology}). 
In the following we say that $P$ and $X_{a,b}$ have the same \emph{geometry} if and only if $\mathcal{B}_1(x_i) =\mathcal{B}_1(x'_i)$,  $\mathcal{B}_2(x_i) = \mathcal{B}_2(x'_i)$, and  $\mathcal{A}(x_i) = \mathcal{A}(x'_i)$ for all $i=1,\ldots,60$.

  \begin{figure}[H] 
\begin{center}
\includegraphics[scale=1.2]{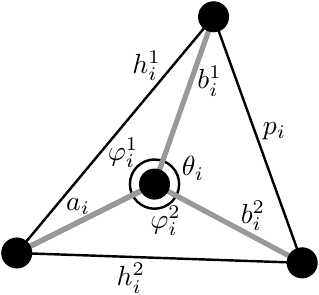}
\caption{Angles, bonds (thick gray lines), and second neighbors (thin black lines) at an atom $x'_i\in \mathcal{P}(X_{a,b})$.} 
\label{topology}
\end{center}
\end{figure}

We are now in the position of stating our main result. 
\begin{theorem}[Geometry of $C_{60}$]\label{maintheorem} Under the
  above assumptions on the potentials $v_{\rm bond}$, $v_{\rm angle}$, $v_{\rm kink}$,
  $v_{\rm dihedral}$, and $v_{\rm  nbd }$ the following holds:
 \begin{enumerate} 
\item[\rm 1.]
{\rm (Minimality in $\mathcal X$)}  For all $\bfc$ and for $\eta$ small enough there exists a
unique minimizer $X_{a^*,b^*}$ of $E_\bfc$ in $\mathcal{X}$. All
configurations in ${\mathcal X}$ different from $X_{a^*,b^*}$ are not
locally stable.  
 \item[\rm 2.] {\rm ( Nonbonded-interaction  effects)} If  $c_{\rm  nbd }=0$, then $a^* = b^*=1$. If $c_{\rm  nbd }=1$ and
\be{twobondscond}
\sqrt{3} v'_{\rm  nbd }(\sqrt{3})\neq  \sigma v'_{\rm  nbd }(\sigma) \ \ \text{with  $\sigma = 2\sin(3\pi/10)$,}
\ee
then  $a^*,b^* \le 1$, $a^* \neq b^*$, and ${\rm sgn}(b^* {-} a^*) = {\rm sgn}\big( \sqrt{3} v'_{\rm  nbd }(\sqrt{3}) {-} \sigma v'_{\rm  nbd }( \sigma)\big).$
 \item[\rm 3.] {\rm (Local stability)} Let $\eta$ be small enough.
\begin{enumerate}
 \item[\rm 3.1]  If $c_{\rm  nbd }=0$ and $c_{\rm angle}\vee  c_{\rm
     kink} = 1$, then $X_{1,1}$ is locally stable with respect to  perturbations in $\mathcal{P}(X_{1,1})$.
\item[\rm 3.2] If $c_{\rm kink}= c_{\rm  nbd }=1$ and $c_{\rm angle}=
  c_{\rm dihedral}=0$, then $X_{a^*,b^*}$  is locally stable with respect to  perturbations in $\mathcal{P}(X_{a^*,b^*})$.
\item[\rm 3.3] If $c_{\rm dihedral}= c_{\rm  nbd }=1$, and $c_{\rm angle}
  \vee c_{\rm kink}=1$, and $\eta'$ and $\eta/\eta'$ are small enough, the configuration $X_{a^*,b^*}$  is locally stable with respect to  perturbations in $\mathcal{P}(X_{a^*,b^*})$.
 \end{enumerate}
\item[\rm 4.] {\rm (Instability)} 
\begin{enumerate}
\item[\rm 4.1]  If $\bfc=(0,0,0,0)$, then there exists $P_1\in\mathcal{P}(X_{a^*,b^*})$ whose geometry differs from $X_{a^*,b^*}$ such that $E_\bfc(P_1)=E_\bfc(X_{a^*,b^*})$. 
\item[\rm 4.2] If  $\bfc=(0,0,0,1)$ or $\bfc=(1,0,0,1)$  and $v_{\rm
    angle}'(2\pi/3) =0$, $v'_{\rm  nbd }(x)>0$ for $x \in I_{\rm  nbd }$, then there exists $P_2\in\mathcal{P}(X_{a^*,b^*})$ such that $E_\bfc(P_2)<E_\bfc(X_{a^*,b^*})$.
 \end{enumerate}
  \end{enumerate}
\end{theorem} 

Depending on which term is active in $E_\bfc$, Theorem
\ref{maintheorem} asserts that
different situations may occur, see Table \ref{table}. By including
 nonbonded-interaction  effects into the picture ($c_{\rm  nbd }=1$) as well as a
 a kink-angular  or a dihedral term ($c_{\rm kink} \vee c_{\rm
   dihedral}=1$), the unique minimizer $X_{a^*,b^*}$ in $\mathcal X$ is locally
 stable and has two different bond lengths under the generic condition
 \eqref{twobondscond}. This corresponds to the actual geometry
 of the $C_{60}$ molecule and is our main result.  As already mentioned in the Introduction,  two distinct bond lengths  $a^* < b
^*$ are experimentally  observed. This is reflected in Assertion 2
 of Theorem \ref{maintheorem}. Indeed,  for  $v_{\rm  nbd }$
convex and increasing in $I_{\rm  nbd }$ we have that  $t \mapsto
t v_{\rm  nbd }'(t)$  is increasing as well, so that $a^* < b
^*$ .
 
If  nonbonded interactions  are neglected ($c_{\rm  nbd }=0$) and either angle term is present, namely $c_{\rm angle}=1$ or
$c_{\rm kink}=1$, the configuration $X_{1,1}$ is stable instead. These cases, already addressed
in \cite{Mainini-Stefanelli12,Stefanelli16}, are unsatisfactory as
they fail to deliver the correct geometry of $C_{60}$, featuring 
indeed two different bond lengths. This  shortcoming  was the main motivation for
the present study.  Let us however stress that the extension of the
argument of \cite{Mainini-Stefanelli12,Stefanelli16} to the case of
nonbonded interactions is nontrivial, as commented in Section~\ref{stability section} below. 

Finally, 
by neglecting both kink-angular and dihedral terms ($c_{\rm kink} = c_{\rm
   dihedral}=0$) no icosahedral configuration in $\mathcal X$ is
 locally stable.  Indeed, we provide an explicit perturbation
   $P_2$ lowering the energy, which consists in simultaneously moving
   the vertices of a pentagonal facet towards the center of the cage
    so to reduce  the length of second neighbors.

Our result focuses on the case where  $\eta$ and $\eta'$ are small,
reflecting indeed that the terms of $E_{\rm  nonbond }$ and $E_{\rm
  dihedral}$  can be supposed to be of lower order with respect
to the two- and three-body part of the energy
\cite{Allinger}. In case both dihedral and  nonbonded-interaction  terms are
present (Assertion 3.3  of Theorem \ref{maintheorem}) we additionally assume $\eta/\eta'$ to be
small, namely that the dihedral term dominates, which
again is well-motivated by the basic chemistry of covalent
bonding in carbon. If this is not the case, for $\bfc =(1,0,1,1)$ the same perturbation $P_2$
of Assertion 4.2  of Theorem \ref{maintheorem}  proves that the the
configuration $X_{a^*,b^*}$ is not locally
stable.

\begin{remark}\label{remark0}
An alternative,
equivalent approach would have been that of considering a single
$v_{\rm two\text{-}body}:=[0,\infty)\to[-1,\infty)$ for all two-body effects,
namely for both first- and second-neighbors, instead of using the two potentials
$v_{\rm bond}$ and $v_{\rm  nbd }$. More precisely, one could introduce
the energy term ${E}_{\rm two\text{-}body}$ defined by 
\begin{equation}\label{twobodytogether}
{E}_{\rm two\text{-}body}(X):=\frac12\sum_{(i,j)\in\,\mathcal{N}_1\cup\mathcal{N}_2}v_{\rm two\text{-}body}(|x_i-x_j|).
\end{equation}
Then, the statements of Theorem \ref{maintheorem} with $c_{\rm  nbd }=1$ can be reformulated
by replacing  ${E}_{\rm bond} +  {E}_{\rm  nonbond }$ by  ${E}_{\rm two\text{-}body}$ and letting 
$$v_{\rm two\text{-}body} =\psi_{\rm bond} v_{\rm bond} + \eta \psi_{\rm  nbd }
v_{\rm  nbd }$$
 where  $\psi_{\rm bond}$ and $\psi_{\rm  nbd }$ are suitable cut-off functions
 supported in $I_{\rm bond}$ and $I_{\rm  nbd }$, respectively. This
 approach would in particular allow to take $v_{\rm two\text{-}body}$
 to have the Lennard-Jones form $v_{\rm two\text{-}body}(r) = kr^{-12} -
 k'r^{-6}$ for suitable positive constants $k$ and $k'$. We however prefer to
 keep our notation as we believe that it delivers a clearer argument.
 \end{remark}
 
\begin{remark}\label{remark} The reference to $sp^2$
  hybridization, that is the assumption that $v_{\rm angle}$ is
  minimized solely at $2\pi/3$ and $4\pi/3$, is here chosen for
  definiteness only. This assumption could be weakened in order to
  encompass more general bonding regimes. In particular, potentials with
  negative slope at $2\pi/3$, such as Brenner-like potentials favoring
  $\pi$ bond angles (see \cite{Brenner90}), could be considered as
  well. In this case, the results correspond to the ones of Theorem
  \ref{maintheorem} along with the choice  $c_{\rm kink} =1$. Moreover, the
  case of $v_3$ having a (small) positive slope at $2\pi/3$ could be
  addressed as well.  Depending on the contribution of $E_{\rm
    dihedral}$ we either get a stability or an instability result as
  in Assertion 3. and Assertion 4. of Theorem
  \ref{maintheorem}. We prefer to present the results under the
  slightly more restrictive assumptions of Section
  \ref{model} for they allow a clearer exposition.

 \end{remark}

\subsection{ Examples of admissible potentials}\label{discussion} 
The assumption frame of Subsection
\ref{model} is sufficiently weak to include virtually all  the specific choices for the potentials which have been
  introduced in the literature \cite{Allinger}. The aim of this subsection is to illustrate
  some concrete examples. In the following, we shall use the indexed symbol $k$ to
  indicate different positive parameters. 

Let us start by observing that the classical potentials
\begin{align*} 
  \text{(harmonic)}\qquad & \frac12\,k_h(r{-}1)^2 -1,\\
\text{(Morse)} \qquad& k_M\left({\rm e}^{-k_M'(r{-}1)}-1\right)^2 -1,\\
\text{(Lennard-Jones)} \qquad&
 \frac{1}{r^{12}}-\frac{2}{r^6}  
\end{align*} 
fulfill assumption \eqref{H2} for $v_{\rm bond}$. The Morse and the
Lennard-Jones potentials, possibly modulated by suitable additional parameters, can
give account of nonbonded interactions as well \cite{Hill,
  Westheimer-Mayer}. In particular, they can be chosen as $v_{\rm nbd}$ and calibrated in
such a way that assumption \eqref{H2s} can be met. In addition, nonbonded interactions
can be described by the classical potential \cite{Hill2}
$$  \text{(Buckingham)} \qquad k_{B}{\rm e}^{-k_{B}'r} -\frac{k''_{B}}{r^6}$$
which again fulfills \eqref{H2s} for a suitable choice of the
coefficients.  Combinations of these potentials can also be considered
in order to model two-body interactions in the spirit of
\eqref{twobodytogether}. 
 
The angle potential $v_{\rm angle}$ is usually defined to be quadratic
around $2\pi/3$ and $4\pi/3$, which ideally fits with assumption
\eqref{H3} but not with \eqref{kink}. Note however that the latter, as
already commented above,  
has no direct chemical justification.
 
Various different formulations for the dihedral term appear in the
literature.  We mention the Molecular Dynamics computational libraries AMBER
\cite{Weiner81}, CHARMM \cite{Brooks83},
GROMOS \cite{Gunsteren87}, Tripos 5.2 \cite{Clark89},  DREIDING \cite{Mayo90}, and  AIREBO \cite{stuart-etal} and refer also to \cite[Subsection 3.2.2]{RC} for
the detailed geometric account of different choices. 
 
In AMBER the contribution of the atom $x_i$ with bonds $\mathcal{B}_1(x_i)=\{a_i,b_i^2,b_i^2\}$ to the dihedral term $v_{\rm dihedral}$ is
\[
v_{\rm AMBER}(\gamma_i):=k_{\rm AMBER}\left[1-\cos\left(3(\gamma_i-\pi)\right)\right],
\]
where $\gamma_i$ is the angle formed by the two planes $\pi_i^j$, containing the bonds with lengths $a_i$ and $b_i^j$, respectively for $j=1,2$ \cite{RC} (see Figure \ref{inversiongamma}).


\begin{figure}[!hpt] 
\begin{center}
\includegraphics[scale=1.2]{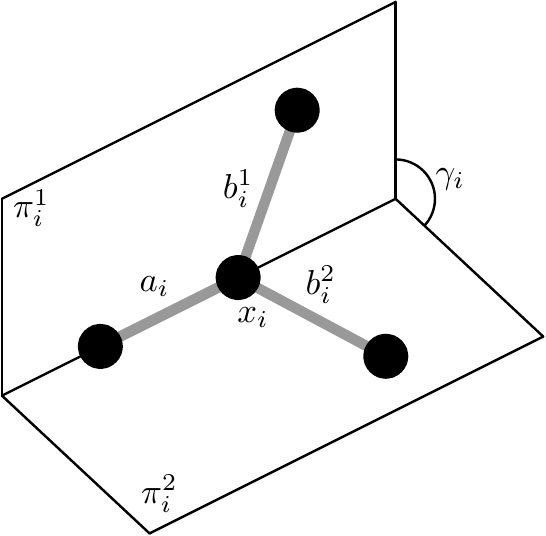}
 \caption{The dihedral potential employed in AMBER depends at every atom $x_i$ on the corresponding angle $\gamma_i$ shown in the picture.}
\label{inversiongamma}
\end{center}
\end{figure} 
In \cite{haddon,Shen-Li} yet another
definition for the dihedral term is introduced. This is based on the
notion of $\pi$-orbital axis vector (POAV), namely the axis
which forms equal angles with the three covalent bonds centered at a given
atom (see \cite[Appendix]{haddon-scott} for a detailed
definition). Both the AMBER and the POAV dihedral terms can be
proved to satisfy  assumption~\eqref{v4hp}.  For example, in the AMBER case we observe that if  $\varphi_i^j=\varphi_i$ for $j=1,2$, then
\[
\sin\frac{\gamma_i}{2}\sin\varphi_i=\sin\frac{\theta_i}{2}
\]
as computed in \cite[Proposition 3.2]{MMPS1}, and hence, assumption \eqref{v4hp} follows from  verifying that 
\[
\frac{\rm d }{{\rm d} \gamma} v_{\rm AMBER}(\gamma)\Big|_{\gamma=\gamma_0}<0
\]
for $\gamma_0\in( 2\pi/3,\pi)$ such that
\[
\sin\frac{\gamma_0}{2}\sin\frac{2}{3}\pi=\sin\frac{3}{10}\pi.
\]

\section{The geometry of $C_{60}$}\label{convexity}
In this section we  specify  the geometry of  $C_{60}$  by
minimizing the energy in the class of icosahedral configurations $\mathcal{X}$. In particular, we
prove    the first and the second assertion of Theorem
\ref{maintheorem}.

\begin{proof}[Proof of Assertion 1 of Theorem \ref{maintheorem}]
For every vector $\bfc$ the energy  $E_\bfc(X_{a,b})$ coincides (up to an additive constant not depending on $a,b \in I_{\rm bond}$) with  (recall \eqref{second bonds} and \eqref{reduced problem}) 
\begin{align}
E_{\eta}(X_{a,b}) &:= 30 v_{\rm bond}(a) + 60 v_{\rm bond}(b) + c_{\rm  nbd } \eta\big(60  v_{\rm  nbd }(2b\sin(3\pi/10))\notag \\& \ \ \ \ + 120  v_{\rm  nbd }(\sqrt{a^2 + b^2+ab})\big)\notag\\
&  = 30 v_{\rm bond}(a) + 60 v_{\rm bond}(b) + \eta f(a,b)\label{E'}
\end{align}
 for a suitable function $f$ being $C^2$ on $I_{\rm  nbd }
 \times I_{\rm  nbd }$  (see below \eqref{H2s}).  We
 first observe that $E_{\eta}(X_{a,b})$ is strictly convex as a
 function of $a,b \in I_2$.  Indeed, the Hessian reads as $ \DD^2
  E_{\eta}(X_{a,b}) = 30 v_{\rm bond}''(a) \ e_1 \otimes e_1 +
 60 v_{\rm bond}''(b) \ e_2 \otimes e_2 + \eta  \DD^2  f(a,b)$ and the assertion follows from \eqref{H2} for $\eta$ small enough.

Consequently, for  such  $\eta$ small enough  there exist
 a unique minimizer $X_{a^*_{\eta},b^*_{\eta}}$ of $E_{\eta}$ 
  in $\overline{I}_{\rm bond} \times \overline{I}_{\rm bond}$. We
observe that for small $\eta$ we have $a^*_{\eta}, b^*_{\eta} \in 
  I_{\rm bond}$. In fact, as $E_{\eta}$ is a continuous perturbation
of $E_0$, one has $(a^*_{\eta}, b^*_{\eta}) \to (a^*_0,b^*_0)$ as
$\eta \to 0$ and $(a^*_0,b^*_0) = (1,1)$ by \eqref{H2}. Consequently,
$X_{a^*_{\eta},b^*_{\eta}}$ is also the unique  minimizer  of $E_\bfc$
over the family $\mathcal{X}$. For given fixed $\eta$ we drop the
subscript $\eta$  and indicate the minimizer as   the
minimizer $(a^*,b^*)$,  for the sake of notational simplicity.  Note that the strict convexity and $a^*, b^* \in  I_{\rm bond}$ imply
\begin{align}\label{firstorder}
\text{$ \DD  E_{\eta}(X_{a,b}) = 0$ for $a,b \in  I_{\rm bond}$ if and only if $(a,b) = (a^*,b^*)$}.
\end{align}
To conclude the proof of Assertion 1  of Theorem \ref{maintheorem},
 it remains to show that $X_{a,b} =\lbrace x_1,\ldots,x_{60}
\rbrace \in \mathcal{X} \setminus \lbrace X_{a^*,b^*} \rbrace$ is not
locally stable for small perturbations. As $\DD 
E_{\eta}(X_{a,b}) \neq 0$ by \eqref{firstorder}, we find $a',b' \in
 I_{\rm bond}$ with $|a-a'|$, $|b-b'|$ arbitrarily small such
that $E_\bfc(X_{a',b'}) < E_\bfc(X_{a,b})$. It now suffices to observe
that  $X_{a',b'}$ can be realized by a configuration $\lbrace
x_1',\ldots,x'_{60} \rbrace  \in \mathcal{X}$  with $|x_i -
x_i'| < \delta_0$ for $i=1,\ldots, 60$.  Indeed,  this
  corresponds to moving  the facets of the pentagons and hexagons
  (infinitesimally) inwardly or outwardly without changing the 
  bond  angles.
\end{proof}

  In the following let  $a^*, b^* \in I_{\rm bond}$   be the length of the bonds of $X_{a^*,b^*}$, and denote by  
\begin{align}\label{ph60}
p^*:=2b^*\sin(3\pi/10)= b^* \sigma, \ \ \ \ h^*:=\sqrt{(a^*)^2+(b^*)^2 +a^*b^*}
\end{align}
the lengths between its second neighbors, where $\sigma = 2\sin(3\pi/10)$.  We now prove the second assertion of Theorem \ref{maintheorem}.
 
 \begin{proof}[Proof of Assertion 2 of Theorem \ref{maintheorem}]
First, if $c_{\rm  nbd }=0$ or, equivalently,  $\eta=0$ in
$E_{\eta}$ (see \eqref{E'}), we have already proved in
\eqref{firstorder} that $(a^*,b^*) = (1,1)$. Suppose now that $c_{\rm
   nbd }=1$ and $\eta>0$.  By computing  the
derivative of $E_{\eta}(X_{a,b})$  we  obtain 
\begin{align}\label{firstorder2}
& \DD  E_{\eta}(X_{a,b}) = \Big(30 v_{\rm bond}'(a) + 120\eta \displaystyle\dfrac{2a+b}{2h} v_{\rm  nbd }'(h),\\ & \quad\quad\quad\quad\quad\quad\quad\quad\quad 60 v_{\rm bond}'(b) + 60\eta \sigma  v'_{\rm  nbd }(\sigma b)+  120\eta \dfrac{a+2b}{2h} v_{\rm  nbd }'(h) \Big),\nonumber
\end{align}  
where, similarly as above, we have set $h = \sqrt{a^2+b^2+ab}$. 
  Since by the assumption on $v_{\rm  nbd }$ (see below
  \eqref{H2s})  we have $0 \le v'_{\rm  nbd }(p^*),  v_{\rm
   nbd }'(h^*)$, \eqref{firstorder} implies  that  $$v_{\rm bond}'(a^*), v_{\rm bond}'(b^*) \le 0$$ and thus $a^*,b^* \le 1$ by  \eqref{H2}.

Assume that \eqref{twobondscond} holds  and observe that this implies
$\sqrt{3} v'_{\rm  nbd }(\sqrt{3} d) \neq  \sigma v'_{\rm 
  nbd }(\sigma d)$ for all $d\in I_{\rm bond}$ if we choose the
neighborhood $I_{\rm bond}$ small enough (depending on $v_{\rm 
  nbd }$). Now suppose $(a^*,b^*) = (d,d)$ for some $d \in I_{\rm
  bond}$. Then \eqref{firstorder} yields $ \DD  E_{\eta}(X_{d,d}) = 0$ and thus by the previous computation we get
\begin{align*}
&30v_{\rm bond}'(d) + 120\eta \dfrac{\sqrt{3}}{2}v_{\rm  nbd
  }'(\sqrt{3}d) \\
&= 60v_{\rm bond}'(d)  + 60\eta \sigma  v'_{\rm
   nbd }(\sigma d)+ 120\eta \dfrac{\sqrt{3}}{2}v_{\rm  nbd
  }'(\sqrt{3}d) = 0,
\end{align*}
which leads to  $\sqrt{3} v'_{\rm  nbd }(\sqrt{3} d) = \sigma
v'_{\rm  nbd }(\sigma d)$. This contradicts assumption
\eqref{twobondscond} and  eventually   shows  that 
$a^* \neq b^*$. Finally, again  by   using the first order optimality condition we derive  
$$v_{\rm bond}'(b^*) - v_{\rm bond}'(a^*) = 3 \eta \dfrac{a^*}{h^*}v_{\rm  nbd }'(h^*)- \eta \sigma v'_{\rm  nbd }(\sigma b^*).$$
As $v'_{\rm bond}$ is increasing  on $I_{\rm bond}$ by \eqref{H2}, we get 
$${\rm sgn} (b^*{-} a^*) = {\rm sgn} \Big( 3 \dfrac{a^*}{h^*}v_{\rm  nbd }'(h^*) - \sigma v'_{\rm  nbd }(\sigma b^*)\Big).$$
Note that we can assume that $|a^*-1|$, $|b^* -1|$, $|h^* - \sqrt{3}|$
are arbitrarily small by simply choosing the neighborhood $I_{\rm
  bond}$ sufficiently small. Consequently, by the regularity of the
potentials the term on the  above right-hand  side has the same sign as $  \sqrt{3} v'_{\rm  nbd }(\sqrt{3}) - \sigma v'_{\rm  nbd }( \sigma)$. This concludes the proof. 
\end{proof}

\section{Stability of $C_{60}$  with  the kink or
  the dihedral term}\label{stability section}

The section is devoted to the proof of the stability results  of
  Assertion 3 of Theorem \ref{maintheorem}. We  follow the general
strategy proposed in \cite[Theorem 7.3]{Mainini-Stefanelli12} which is
based on convexity and monotonicity arguments for the energy
$E_\bfc$. In our context, however, a more elaborated analysis of the
properties of the  phenomenological energy $E_\bfc$ is 
required. Indeed, the original argument in
\cite{Mainini-Stefanelli12} is based on
the possibility of treating bonded and angle effects {\it
  separately}. This is here not possible, as both first-neighbor bond lengths and
angles contribute to the length of second neighbors, hence to the term
$E_{\rm nonbonded}$. The strategy is hence to exploit the smallness of
$E_{\rm nonbonded}$, that is of $\eta$, in order to keep this
additional intricacy under control.

Let us start by rewriting  the energy corresponding to a small   perturbation $P = \lbrace x_1', \ldots, x'_{60} \rbrace \in \mathcal{P}(X_{a^*,b^*})$ (see \eqref{perturbations}) as a sum
$$
E_\bfc(P)=\sum_{i=1}^{60}\widehat{E}_\bfc(x'_i)
$$ 
where $\widehat{E}_\bfc(x_i')$ is the energy contribution associated to a single atom. In view of \eqref{perturbedbonds}-\eqref{perturbed2bonds2} each $\widehat{E}_\bfc(x'_i)$ can be expressed as a function in terms of the covalent bonds and the angles, i.e. with a slight abuse of notation $\widehat{E}_{\bfc}(a_i,b_i^1,b_i^2,\theta_i,\varphi^1_i,\varphi^2_i)= \widehat{E}_\bfc(x_i)$, and is defined by (set $y_i= (a_i,b_i^1,b_i^2,\theta_i,\varphi^1_i,\varphi^2_i)$ for shorthand)
\begin{align} 
\widehat{E}_\bfc(y_i) & :=\widehat{E}_{\rm bond}(y_i) + c_{\rm
  angle}\widehat{E}_{\rm angle}(y_i) +c_{\rm kink}\widehat{E}_{\rm
  kink}(y_i) \nonumber\\& \ \ \ \ +\eta'\, c_{\rm
  dihedral}\widehat{E}_{\rm dihedral}(y_i) + \eta\,c_{\rm  nbd
  }\widehat{E}_{\rm  nonbond }(y_i) \label{energysplit} 
\end{align}
for every $x_i\in P$, where we have 
\begin{align*}
&\widehat{E}_{\rm bond}(y_i):=\dfrac{1}{2}v_{\rm bond}(a_i)+\dfrac{1}{2}v_{\rm bond}(b^1_i)+\dfrac{1}{2}v_{\rm bond}(b^2_i),  \\ 
&\widehat{E}_{\rm angle}(y_i):=v_{\rm angle}(\theta_i)+v_{\rm angle}(\f^1_i)+v_{\rm angle}(\f^2_i), \\& \widehat{E}_{\rm kink}(y_i):=v_{\rm kink}(\theta_i)+v_{\rm kink}(\f^1_i)+v_{\rm kink}(\f^2_i),\\
&\widehat{E}_{\rm dihedral}(y_i):=v_{\rm dihedral}(\theta_i,\varphi_i^1,\varphi_i^2),\quad\textrm{and}\\
& \widehat{E}_{\rm  nonbond }(y_i):=  v_{\rm  nbd }(p_i)+  v_{\rm  nbd }(h^1_i)+   v_{\rm  nbd }(h^2_i).
\end{align*}
 Note that the factor $1/2$ in the definition of $\widehat{E}_{\rm bond}$
 takes into account the fact that each bond is shared by two atoms. We first use the convexity properties in \eqref{H2} and \eqref{H3} to show the convexity of $\widehat{E}_\bfc$.

\begin{proposition}[ Strict Convexity  of $\widehat{E}_c$]\label{thmcxty}
If   $\eta$, $\eta'$ are taken small enough, then for every vector  $\bfc \in \lbrace 0,1\rbrace^4$ with $ c_{\rm angle} \vee  c_{\rm kink}=1$  the energy $\widehat{E}_\bfc$ is strictly convex on $I_{\rm bond}^3 \times I_{\rm angle}^3$.  If $c_{\rm  nbd } = 0$ or  $c_{\rm dihedral} = 0$, the choice of $\eta$ or $\eta'$, respectively, is arbitrary. \end{proposition}

\begin{proof}
We split $\widehat{E}_\bfc = f_1+f_2 + f_3$ into the parts $f_1 =
\widehat{E}_{\rm bond}$, $f_2 = \eta \, c_{\rm  nbd }
\widehat{E}_{\rm  nonbond } + \eta' \,c_{\rm dihedral}
\widehat{E}_{\rm dihedral}$ and $f_3 =   c_{\rm angle}\widehat{E}_{\rm
  angle} + c_{\rm kink} \widehat{E}_{\rm kink}$.   We consider two
points $y_1, y_2  \in I_{\rm bond}^3 \times I_{\rm angle}^3$ and
distinguish the bond and angle part by writing $y_i = (y^1_i,y_i^2)$
for $i=1,2$ with $y^1_i,y^2_i \in \Rz^3$. We let $ \lambda_{\rm
    bond} = \min_{\ell  \in \overline{I}_{\rm bond}} v_{\rm
  bond}''(\ell)/2$ and by the smoothness of $\widehat{E}_{\rm 
  nonbond }$ and $\widehat{E}_{\rm dihedral}$ we can choose $
  \lambda^* \in \Rz$ such that for each $y \in I_{\rm bond}^3 \times
I^3_{\rm angle}$ the smallest eigenvalue of $ \DD^2  \widehat{E}_{\rm  nonbond }(y)$ and $ \DD^2  \widehat{E}_{\rm dihedral}(y)$ is larger than $\lambda^*$. It is a well known  fact in the theory of convex functions that this implies for $t \in [0,1]$ 
\begin{align*}
f_1(ty_1 + (1{-}t)y_2)&\le t f_1(y_1) + (1{-}t) f_1(y_2) - \dfrac{1}{2}\lambda_{\rm bond} t(1{-}t)|y_1^1 {-} y_2^1|^2, \\ 
f_2(ty_1 + (1{-}t)y_2) &\le t f_2(y_1) + (1{-}t) f_2(y_2) \\
&+ \dfrac{1}{2}
|\lambda^*| (\eta c_{\rm  nbd } {+} \eta'c_{\rm dihedral}) t(1{-}t)|y_1 {-} y_2|^2.
\end{align*}
Moreover, by the strong convexity of $f_3$ in the angle variables (see \eqref{H3}) we have for $t \in [0,1]$ 
$$f_3(ty_1 + (1{-}t)y_2)\le t f_3(y_1) + (1{-}t) f_3(y_2) - \lambda_{\rm angle} t(1{-}t)|y_1^2 {-} y_2^2|^2.$$
Thus, since $\lambda_{\rm bond} >0$ by \eqref{H2} and $\lambda_{\rm angle}>0$, we derive for $\eta,\eta'$ small enough that $$\min\left\{ \frac{1}{2}\lambda_{\rm bond}, \lambda_{\rm angle} \right\}  - \frac{1}{2} |\lambda^*| (\eta c_{\rm  nbd } + \eta'c_{\rm dihedral}) > 0$$  and conclude that $\widehat{E}_\bfc$ is strictly convex.   \end{proof}

We now derive a monotonicity property which can be recovered from
the kink assumption \eqref{kink} or from assumption \eqref{v4hp}. Note
that in both cases the argument is based on the fact that  the
planarity of the faces is energetically favored by   $\widehat{E}_{\rm kink}$ or $\widehat{E}_{\rm dihedral}$.

\begin{proposition}[ Monotonicity  of $\widehat{E}_c$]\label{thmmon}
Assume   $\eta$ is small enough. If 
\begin{align*}
{\rm 1. } \ \ & \text{$\bfc = (1,0,0,0)$ or} \\
{\rm 2. } \ \ &  \text{$\bfc\in \lbrace (0,1,0,1), (0,1,0,0)\rbrace$ or} \\
{\rm 3. } \ \ &  \text{$\bfc \in \lbrace (1,0,1,1), (0,1,1,1), (1,0,1,0), (0,1,1,0) \rbrace$ and  $\eta/\eta'$ small enough,}  
\end{align*}
then we have for all $a,b \in I_{\rm bond}$, $\theta \in \Theta_\eps$,  $\varphi \in \Phi_\eps$ with $\theta \le 3\pi/5$ and $\varphi \le 2\pi/3$
\begin{align}\label{monassertion}
\widehat{E}_\bfc(a,b,b,\theta,\varphi,\varphi) \ge \widehat{E}_\bfc(a^*,b^*,b^*,3\pi/5,2\pi/3,2\pi/3),
\end{align}
where equality only holds if $a=a^*$, $b=b^*$, $\theta = 3\pi/5$ and $\varphi=2\pi/3$. 
   \end{proposition}

\begin{proof}
We first observe that in  Case 1  one has $a^*=b^*=1$ and the assertion follows directly from  \eqref{H2} and \eqref{H3} (cf. also the arguments in \cite[Theorem 7.3]{Mainini-Stefanelli12}). 

For  Cases 2 and 3  we split  $\widehat{E}_\bfc = f_1+f_2$
into the parts $f_1(a,b,\theta,\varphi) = \widehat{E}_{\rm bond}(y) +
\eta\, c_{\rm  nbd } \widehat{E}_{\rm  nonbond }(y)$
and $f_2(\theta,\varphi) = c_{\rm angle}\widehat{E}_{\rm angle}(y) +
c_{\rm kink} \widehat{E}_{\rm kink}(y) + \eta'\,c_{\rm dihedral}
\widehat{E}_{\rm dihedral}(y)$, where for shorthand $y =
(a,b,b,\theta,\varphi,\varphi)$. By the smoothness of $v_{\rm bond}$
and $v_{\rm  nbd }$  and by  \eqref{perturbed2bonds2}
we  can find  a constant $C_1$ independent of $\eta$ such that
$| \DD_{\theta,\varphi}  f_1(a,b,\theta,\varphi)|\le C_1\eta$
for all $a,b \in I_{\rm bond}$ and $\theta \in \Theta_\eps$,  $\varphi
\in \Phi_\eps$. Let for shorthand $\theta_0 = 3\pi/5$ and $\varphi_0
=2\pi/3$. Then we get by  Taylor's formula for all $a,b \in I_{\rm
  bond}$, $\theta \in \Theta_\eps$,    and  $\varphi \in \Phi_\eps$  one has 
\begin{align}\label{mon:f_1}
|f_1(a,b,\theta,\varphi) - f_1(a,b,\theta_0,\varphi_0)| &\le C_1 \eta (|\theta - \theta_0| +|\varphi - \varphi_0|) + C_1(|\theta - \theta_0| +|\varphi - \varphi_0|)^2 \notag \\ & \le C_1(\eta+ \eps) (|\theta - \theta_0| +|\varphi - \varphi_0|)  
\end{align}
passing possibly to a larger constant $C_1$  without introducing
new notation.  We now show that there  exists 
a constant $C_2>0$ such that for all $\theta \in \Theta_\eps$,
$\varphi \in \Phi_\eps$ with $\theta \le \theta_0$,  and 
$\varphi \le \varphi_0$  one has 
\begin{align}\label{mon:f_2}
f_2(\theta,\varphi) - f_2(\theta_0,\varphi_0) &\ge C_2(c_{\rm kink} + \eta'\, c_{\rm dihedral}) ( (\theta_0 -\theta)  + (\varphi_0 - \varphi)).  
\end{align}
 In Case 2  we use the convexity of $v_{\rm kink}$ in $I_{\rm
  angle}$ and condition \eqref{kink} to derive $v_{\rm kink}(\theta) -
v_{\rm kink}(\theta_0) \ge  \lambda(\theta_0 -\theta)$  and
  $v_{\rm kink}(\varphi) - v_{\rm kink}(\varphi_0) \ge  \lambda(\varphi_0 -\varphi)$, where $\lambda := -\lim_{\theta \uparrow 2\pi/3} v'_{\rm kink}(\theta)>0$ by \eqref{kink}. This implies \eqref{mon:f_2}. 
 In Case 3  we first observe that the strong convexity of
$v_{\rm angle}$ assumed in \eqref{H3} and the fact that the minimum
value is attained at $ \theta_0 =  2\pi/3$ implies $v_{\rm angle}(\theta) - v_{\rm angle}(\theta_0) \ge C_3(\theta_0-\theta)$ for a constant $C_3>0$ depending only on $\lambda_{\rm angle}$.  Moreover, the smoothness of $v_{\rm dihedral}$ implies $v_{\rm dihedral}(\theta,\varphi,\varphi) - v_{\rm dihedral}(\theta_0,\varphi,\varphi) \ge -  C_4(\theta_0-\theta)$ for a constant $C_4>0$ large enough. Now we use \eqref{v4hp} and Taylor's formula to compute 
$$v_{\rm dihedral}(\theta_0,\varphi,\varphi) - v_{\rm dihedral}(\theta_0,\varphi_0,\varphi_0) \ge   \lambda'(\varphi_0 - \varphi) -   C_5(\varphi_0 - \varphi)^2 \ge (\lambda' - C_5\eps)(\varphi_0 - \varphi)$$
for $C_5>0$ large enough, where $$\lambda'=-\frac{\rm d }{\rm d
  \varphi} v_{\rm dihedral}(\theta_0, \varphi,\varphi)\Big|_{\varphi =
  2\pi/3}>0$$  by  assumption  \eqref{v4hp}.  Collecting the last estimates and using $v_{\rm angle}(\varphi) \ge v_{\rm angle}(\varphi_0)$ we conclude 
$$f_2(\theta,\varphi) \ge f_2(\theta_0,\varphi_0) + (C_3-C_4\eta')(\theta_0-\theta) + \eta' (\lambda'- C_5\eps)(\varphi_0 - \varphi),$$
which for $\eta'$ and $\eps$ small enough implies \eqref{mon:f_2}. 

We are now in the position to show \eqref{monassertion}.  By
combining  \eqref{mon:f_1} and \eqref{mon:f_2} we derive for $\eta/\eta'$ small and $\eps$ small with respect to $\eta$ and $\eta'$
$$ \widehat{E}_\bfc(a,b,b,\theta,\varphi,\varphi) \ge \widehat{E}_\bfc(a,b,b,\theta_0,\varphi_0,\varphi_0),$$
where equality only holds if $\theta = \theta_0$ and $\varphi = \varphi_0$. Recalling \eqref{E'} and the fact that $(a^*,b^*)$ minimizes $E_{\eta}(X_{a,b})$ we conclude
 $$ \widehat{E}_\bfc(a,b,b,\theta,\varphi,\varphi) \ge  \dfrac{1}{60} E_{\eta}(X_{a,b}) \ge \dfrac{1}{60} E_{\eta}(X_{a^*,b^*}) = \widehat{E}_\bfc(a^*,b^*,b^*,\theta_0,\varphi_0,\varphi_0),$$
 where due to the uniqueness of the minimizer of ${E}_{\eta}$ equality only holds if $a=a^*$ and $b=b^*$. 
\end{proof}

We are now in the position to prove the stability of 
  $X_{a^*b^*}$ under small perturbations. In the following proof we
will treat  Assertions  3.1,  3.2,  and  3.3 of  Theorem \ref{maintheorem} simultaneously.

\begin{proof}[Proof of Assertion 3 of Theorem \ref{maintheorem}]
Let  $P = \lbrace x'_1,\ldots,x'_{60} \rbrace \in \mathcal{P}(X_{a^*,b^*})$ be given and suppose that the assumptions stated in Assertion 3.1,  3.2, or 3.3, respectively, are satisfied.   Recalling \eqref{perturbedbonds}-\eqref{perturbed2bonds2}, to each $x'_i$, $i=1,\ldots,60$, we associate $a_i,b_i^1,b_i^2$,$\theta_i$, $\varphi_i^1,\varphi_i^2$.  Let us assume that $P$ has not the same geometry as $X_{a^*,b^*}$, i.e. not all bond lengths and angles coincide with the corresponding values of $X_{a^*,b^*}$. We  show that then indeed $E_\bfc(P) > E_\bfc(X_{a^*,b^*})$. 

Define the mean values
\begin{align*}
  \bar{a} = \frac{1}{60} \sum_{i=1}^{60} a_i, \quad \bar{b} =
  \frac{1}{120}\sum_{i=1}^{60} (b_i^1 + b_i^2), \quad\bar{\theta} =
  \frac{1}{60} \sum_{i=1}^{60} \theta_i, \ \text{and}  \ \bar{\varphi} =
  \frac{1}{120}\sum_{i=1}^{60} (\varphi_i^1 + \varphi_i^2).
\end{align*}
Then we apply Proposition \ref{thmcxty} for $\eta$ and $\eta'$ small enough (if $c_{\rm  nbd } = 1$ or $c_{\rm dihedral} = 1$, respectively)  and use twice the strict convexity of $\widehat{E}_\bfc$ to obtain
\begin{align}
E_\bfc(P) &= \sum^{60}_{i=1}\widehat{E}_\bfc\left(a_i,b_i^1,b_i^2,\theta_i, \varphi_i^1,\varphi_i^2\right) \nonumber\\ 
& \ge \sum^{60}_{i=1}\widehat{E}_\bfc \left(a_i,\dfrac{b_i^1 +
    b_i^2}{2},\dfrac{b_i^1 + b_i^2}{2},\theta_i,
  \dfrac{\varphi_i^1+\varphi_i^2}{2},\dfrac{\varphi_i^1+\varphi_i^2}{2}\right) \nonumber
\\ & \ge 60\,
\widehat{E}_\bfc\left(\bar{a},\bar{b},\bar{b},\bar{\theta},\bar{\varphi},\bar{\varphi}\right), \label{stabil1}
\end{align}
where we have equality if and only if each bond and each angle coincides with the corresponding mean value. As the five angles of each pentagon sum up at most to $3\pi$, we obtain $\bar{\theta} \le 3\pi/5$. A similar argument for the angles of a hexagon, whose sum does not exceed $4\pi$, yields $\bar{\varphi} \le 2\pi/3$. We observe that for each vector $\bfc$ in the  Assertion 3.1-3.3 one of the Assumptions 1-3 of Proposition \ref{thmmon} is satisfied. Consequently, we obtain 
\begin{align}\label{stabil2}
60 \widehat{E}_\bfc(\bar{a},\bar{b},\bar{b},\bar{\theta},\bar{\varphi},\bar{\varphi}) \ge 60\widehat{E}_\bfc(a^*,b^*,b^*,3\pi/5,2\pi/3, 2\pi/3) = E_\bfc(X_{a^*,b^*}),
\end{align}
where equality only holds if $\bar{a} = a^*$, $\bar{b} = b^*$, $\bar{\theta}= 3\pi/5$ and $\bar{\varphi} = 2\pi/3$. As by assumption $P$ has not the same geometry as $X_{a^*,b^*}$,  \eqref{stabil1}-\eqref{stabil2} yield $E_\bfc(P)>E_\bfc(X_{a^*,b^*})$. 
\end{proof}

\section{Nonstability results}\label{counterexample section}


 In this section we establish Assertion 4 of Theorem
 \ref{maintheorem}.  After translation and rotation of 
   $X_{a^*,b^*} = (x^*_1,\ldots,x_{60}^*)$ we may assume that $x^*_i
 \in \Rz \times \Rz \times [0,\infty)$ for all indexes
 $i=1,\ldots,60$, and  that  one pentagon of $X_{a^*,b^*}$ is contained in $\Rz \times \Rz \times \lbrace 0 \rbrace$ with vertices 
 $$V_i:=(\cos(2i\pi/5), \sin(2i\pi/5),0)$$
  for $i=1,\ldots,5$. Let us relabel the atoms of $X_{a^*,b^*}$ so that, for $i=1,\ldots,5$, we have $x^*_i := V_i$, and, for $i=6,\ldots,10$, the atom $x^*_{i}$ is the only neighboring atom of  $x^*_{i-5}$ not lying in $\lbrace V_1,\ldots,V_5 \rbrace$.  By $H^*_1, \ldots, H^*_5$ we denote the planes in $\Rz^3$ containing the five hexagonal faces of $X_{a^*,b^*}$ adjacent to the pentagon formed by $\lbrace V_1,\ldots, V_5 \rbrace$.


The perturbation $P_1 =  \lbrace x'_1,\ldots,x'_{60} \rbrace $ is defined by  setting $x'_i := x_i^*$ for $i \ge 6$ and 
$$x'_i := (\cos(2i\pi/5 + t_1), \sin(2i\pi/5+t_1),t_2)$$
 for $i=1,\ldots,5$ and for some positive (small) constants $t_1$ and $t_2$ to be specified, i.e.,  the transformation rotates one of the twelve pentagonal faces of the molecule.

\begin{proof}[Proof of Assertion 3.1 of Theorem \ref{maintheorem}]
First we see that for $t_1,t_2$ sufficiently small $P_1 \in
\mathcal{P}(X_{a^*,b^*})$. Moreover, the geometry of $P_1$ and
$X_{a^*,b^*}$ are clearly different as, e.g., the hexagons adjacent to
the pentagon formed by $x'_1,\ldots,x'_5$ are not planar. Recall that
by Assertion 2 of Theorem  \ref{maintheorem} we have  $X_{a^*,b^*}
  = X_{1,1}$. Consequently, to prove $E_\bfc(P_1) = E_\bfc(X_{1,1})$
for $\bfc = (0,0,0,0)$,  it  suffices to show that each bond has length $1$. 

We observe  that the only bonds that can present a different length
in $P_1$ with respect to $X_{a^*,b^*}$ are the ones in
$\bigcup^5_{i=1} \mathcal{B}_1(x'_i)$. As the pentagon contained in
$\Rz \times \Rz \times \lbrace 0 \rbrace$ is  just  rotated, we find $|[x_1,x_2]|= |[x_2,x_3]|= \ldots = |[x_5,x_1]| = 1$. Finally, for a suitable choice of $t_2$ with respect to $t_1$ one can additionally obtain  $|[x_i,x_{i+5}]| = 1$ for $i=1,\ldots,5$. This concludes the proof. 
\end{proof}

The definition of perturbation $P_2$ is more involved. One pentagonal face of $X_{a^*,b^*}$ is moved in  such a way that the length of the bonds in $\bigcup_{i=1}^{60}\mathcal{B}_1(x_i^*)$  shared by two hexagons do not change to  first order. More precisely,  for each $i=1,\ldots,5$ we find a unique vector 
$$v_i = (a \cos(2i\pi/5),a\sin(2i\pi/5),b) \in \Rz^3$$
 for suitable  constants  $a,b\in(0,1)$ such that
\begin{align}\label{atomperturb}
|v_i|=1\quad\textrm{and}\quad v_i \cdot (x^*_i - x^*_{i+5}) =0.
\end{align} 
Then we define $P_2 = \lbrace x'_1,\ldots,x'_{60} \rbrace$ by  setting $x'_i := x_i^*$  for $i \ge 6$ and  $x'_i := x_i^* + tv_i$  for $i=1,\ldots,5$  and for  a small constant  $t >0$.  By \eqref{atomperturb},  each segment $[x_i^*, x_i'] \subset \Rz^3$ is not contained in the planes $H_j^*$, $j=1,\ldots,5$, and therefore   the five hexagons of $P_2$ adjacent to the pentagon $\lbrace x'_1, \ldots,x'_5  \rbrace$ are not planar, but each one is kinked along the corresponding segment with endpoints in $\lbrace x'_6,\ldots,x'_{10} \rbrace$. Observe that the essential point of the transformation $P_2$   is the nonplanarity of these hexagons since hereby (i) the length of the second neighbors can be reduced and (ii)  the energy increase due to the modification of the angles is negligible since $v'_{\rm angle}(2\pi/3) = 0$.

\begin{proof}[Proof of Assertion 3.2 of Theorem \ref{maintheorem}]
   Let $P_2$ be defined as above and note that for $t$
  sufficiently small we have that  $P_2 \in
  \mathcal{P}(X_{a^*,b^*})$. We now proceed in two steps. In Step I we
  analyze the modification of bond lengths and angles induced by the
  perturbation. In Step II, we calculate the energy difference of the
  two configurations $X_{a^*,b^*}$ and $P_2$, which in view of the
  first order optimality condition derived in \eqref{firstorder}
  depends only on the angles (see \eqref{energydiff} below). We 
  will  then  be able to   conclude since the sum over all
  angles  strictly  decreases due to the nonplanarity of 
  the  five hexagons (see \eqref{differentangles} below). 

\emph{Step I.}  The transformation  changes only
 the bond lengths $\bigcup_{i=1}^5 \mathcal{B}_1(x_i)\cup
\mathcal{B}_2(x_i)$, the  angles $\bigcup_{i=1}^{10}
\mathcal{A}(x_i)$, and  the second neighbors $\bigcup_{i=6}^{10}
\mathcal{B}_2(x_i)$.  In particular,   there  exist  $\delta_b, \delta_\varphi \in \Rz$ such that for the covalent bonds and the angles associated to $x'_1,\ldots,x'_5$ we have for $i=1,\ldots,5$
\begin{subequations}\label{nonstabil1}
\begin{align}
 a_i &= a^* +{\rm O}(t^2)\label{a}\\
 b^j_i &= b^* + t\delta_b + {\rm O}(t^2), \ j=1,2, \\
\theta_i &= 3\pi/5,\\
  \varphi_i^j &= 2\pi/3 + t\delta_{\varphi} + {\rm O}(t^2).\label{phi}
\end{align}
\end{subequations}
Note that  due to the symmetry of the transformation  all these
quantities are  actually independent of $i$ and $j$ and the
pentagon $\lbrace x'_1,\ldots,x'_5 \rbrace$ is regular and planar.
Moreover,  \eqref{a}  follows from the fact that  the length
of the  bonds  shared by two hexagons do not change in first
order by \eqref{atomperturb}. Likewise, for the second-neighbors an elementary computation yields in view of  \eqref{perturbed2bonds2}, \eqref{ph60}, \eqref{nonstabil1},  and $\cos'(2\pi/3)= -\sqrt{3}/2$ for $i=1,\ldots,5$, $j=1,2$ (we again set $\sigma = 2 \sin(3\pi/10)$ for shorthand)
\begin{align*}
 p_i &= \sqrt{2(1-\cos(\theta_i))}(b^* + t\delta_b + {\rm O}(t^2))  = p^* + t\delta_b  \sigma  + {\rm O}(t^2),\\
  h_i^j &= \sqrt{a_i^2 + (b_i^j)^2 - 2a_i b_i^j \cos(\varphi_i^j)} \\&= \sqrt{(a^*)^2 + (b^*)^2 + a^* b^* + 2t \delta_b b^* - 2ta^*  (\cos(2\pi/3)\delta_b- b^*\sqrt{3}\delta_\varphi/2) + {\rm O}(t^2)} \nonumber\\
& = h^* + \dfrac{t}{h^*} \left( \delta_b\left(b^*  + a^*/2\right) +
  \sqrt{3}a^*b^* \delta_\varphi/2 \right) +  {\rm O}(t^2). 
\end{align*}
Moreover, for $i=6,\ldots,10$, $j=1,2$, we find for $\delta_\varphi'
\in \Rz$ by a similar computation
\begin{subequations}
  \label{nonstabil2}
  \begin{align}
    a_i &= a^*,  \ \ b_i^j = b^*,  \ \ \theta_i = 3\pi/5, \ \
    p_i = p^*,\\
\varphi^j_i &= 2\pi/3 + t\delta_\varphi' + {\rm O}(t^2),\label{phi2}\\
    h_i^j &= \sqrt{a_i^2 + (b_i^j)^2 - 2a_i b_i^j \cos(\varphi_i^j)} = \sqrt{(a^*)^2 + (b^*)^2 - 2a^* b^* \cos(\varphi_i^j)} \nonumber\\
    &= \sqrt{(a^*)^2 + (b^*)^2 + a^* b^* +  \sqrt{3}ta^* b^* \delta'_\varphi) + {\rm O}(t^2)} \nonumber\\
    & = h^* + \sqrt{3}t/(2h^*) a^*b^* \delta'_\varphi + {\rm O}(t^2).
  \end{align}
\end{subequations}

We close the discussion about the modification of bonds and angles by showing 
\begin{align}\label{differentangles}
\delta_\varphi + \delta_\varphi' <0.
\end{align}
To see this,  we recall that the five hexagons adjacent to
  $\lbrace x_1',\ldots,x_5' \rbrace$ are kinked along the
  corresponding segment with endpoints in $\lbrace x'_6,\ldots,x'_{10}
  \rbrace$. Each hexagon consists of two planar quadrangles
$Q^1_i,Q^2_i$ with angles $2\pi/3,2\pi/3, \pi/3, \pi/3 $ and
$\varphi,\varphi, \pi-\varphi, \pi- \varphi $ with
$\varphi:=\varphi_1^1$ as given in  \eqref{phi}.

 We have already noticed below \eqref{atomperturb} that the
  segment $[x_i^*, x_i']$  is not contained in  $H_j^*$,
  $i,j=1,\ldots,5$. Consequently, the angle enclosed by the the two
planes containing $Q^1_i$ and $Q^2_i$, respectively, is larger than
$Ct$ for a sufficiently small universal constant $C>0$. By an
elementary trigonometric argument this implies that $\varphi^j_i$ is
smaller than the sum of the two corresponding angles $\pi- \varphi$,
$\pi/3$ of the quadrangles at $x'_i$. More precisely, $\varphi^j_i \le
(\pi - \varphi ) +\pi/3  - C't$ for $C'>0$ small enough for all
$i=6,\ldots,10$, $j=1,2$. Therefore, in view of  \eqref{phi} and
\eqref{phi2}  we derive 
$$2\pi/3 + t\delta_\varphi' + {\rm O}(t^2) \le 2\pi/3 - t\delta_\varphi - C't + {\rm O}(t^2)$$
and see that \eqref{differentangles} holds true since $t>0$.

\emph{Step II.}  We now estimate the difference of
$E_\bfc(X_{a^*,b^*})$ and $E_\bfc(P_2)$.  Let us start from the
case   $\bfc=(1,0,0,1)$. By  \eqref{nonstabil1}, \eqref{nonstabil2} and the fact that $v_{\rm angle}'(2\pi/3)=0$ we get $v_{\rm bond}(a_i) = v_{\rm bond}(a^*) + {\rm O}(t^2)$ and $v_{\rm angle}(\varphi_i^j) = v_{\rm angle}(2\pi/3) + {\rm O}(t^2) = {\rm O}(t^2)$ for $i=1,\ldots,10$, $j=1,2$.   Consequently, recalling \eqref{energysplit} and \eqref{nonstabil1}-\eqref{nonstabil2} we obtain $E_\bfc(P_2) - E_\bfc(X_{a^*,b^*}) = A + B +  C$ with
\begin{align*}
A&:=  \sum^{5}_{i=1} \left(\widehat{E}_{\rm bond}(x_i) -  \widehat{E}_{\rm bond}(x^*_i)\right) = \sum^{5}_{i=1} \sum^{2}_{j=1} \left(\dfrac{1}{2}v_{\rm bond}(b^j_i) - \dfrac{1}{2}v_{\rm bond}(b^*)\right)+ {\rm O}(t^2) \\& \quad\quad\quad\quad\quad\quad\quad\quad\quad\quad\quad\quad\quad \ \quad=  5v_{\rm bond}'(b^*)t\delta_b  + {\rm O}(t^2), \\
B&:= \sum^{10}_{i=1} \left(\widehat{E}_{\rm angle}(x_i) -  \widehat{E}_{\rm angle}(x^*_i)\right) =  {\rm O}(t^2), \\
C&:= \sum^{10}_{i=1} \eta\left(\widehat{E}_{\rm  nonbond
    }(x_i) -  \widehat{E}_{\rm  nonbond }(x^*_i)\right) \\
&= 5 v'_{\rm  nbd }(p^*)\eta t\delta_b \sigma + 10v'_{\rm  nbd }(h^*) \eta \sqrt{3}t/(2h^*)a^*b^* \delta'_\varphi  \\ & \quad\quad\quad\quad\quad\quad\quad \quad\quad\quad\quad\quad + 10v'_{\rm  nbd }(h^*)\eta \dfrac{t}{h^*} \left(  \delta_b(b^*  + a^*/2) + \sqrt{3}a^*b^* \delta_\varphi/2 \right) + {\rm O}(t^2).
\end{align*}
In case $\bfc=(1,0,0,0)$ we obtain  the same estimate with
$B=0$.  Recall that we have identified the geometry of
$X_{a^*,b^*}$ by optimizing the energy $E_{\eta}(X_{a,b})$ defined in
\eqref{E'} in terms of $a,b$. In particular, the first-order
optimality condition $$ \frac{\rm d}{{\rm d} b}E_{\eta}(X_{a^*,b})\Big|_{b
  = b^*}  =
0$$ 
yields  (cf. \eqref{firstorder2})  
$$v_{\rm bond}'(b^*) + \eta \sigma v'_{\rm  nbd }(p^*)+  2\eta \dfrac{a^*+2b^*}{2h^*} v_{\rm  nbd }'(h^*)  = 0.$$
Combining the last two estimates we derive
\begin{align}\label{energydiff}
E_\bfc(P_2) - E_\bfc(X_{a^*,b^*})  =  10v'_{\rm  nbd }(h^*) \eta \sqrt{3}t/(2h^*)  a^*b^* (\delta_\varphi + \delta'_\varphi) + {\rm O}(t^2).
\end{align}
In view of \eqref{differentangles} and the fact that $v'_{\rm  nbd }(h^*)>0$ and $t>0$ we conclude that   $E_\bfc(P_2) - E_\bfc(X_{a^*,b^*}) < 0$. 
\end{proof}

\section*{Acknowledgements}
This work has been funded by the Vienna Science
and Technology Fund (WWTF) through project MA14-009. Partial support 
 by the Austrian Science Fund (FWF) through grants P 27052 and I
 2375-N32.  The Authors would like to acknowledge
the kind hospitality of the Erwin
Schr\"odinger International Institute for Mathematics and Physics,
where part of this research was developed under the frame of the Thematic
Program {\it Nonlinear Flows}. 


\bibliographystyle{alpha}

\end{document}